\documentclass[12pt]{article}
\usepackage{amsfonts,amsmath,amsthm,color,amssymb,graphicx,subfigure,dsfont}
\usepackage{algorithm,enumerate}
\usepackage{multirow, booktabs}
\usepackage{xpatch} 
\usepackage{hyperref}
\usepackage{threeparttable}

\newcommand{\blind}{0}
\allowdisplaybreaks
\usepackage[authoryear, sort&compress]{natbib} 

\def\boxit#1{\vbox{\hrule\hbox{\vrule\kern6pt
\vbox{\kern6pt#1\kern6pt}\kern6pt\vrule}\hrule}}

\newcommand{\mbf}{\mathbf}
\newcommand{\sbf}{\boldsymbol}

\newtheorem{assumption}{Assumption}  %
\newtheorem{lemma}{Lemma}[section]            %
\newtheorem{theorem}{Theorem}[section]        %
\newtheorem{corollary}{Corollary}[section]    %
\newtheorem{proposition}{Proposition}[section] %
\numberwithin{equation}{section}    %
\makeatletter
\xpatchcmd{\@thm}{\thm@headpunct{.}}{\thm@headpunct{:}}{}{}

\makeatletter
\newcommand*{\rom}[1]{\expandafter\@slowromancap\romannumeral #1@}
\makeatother

\def\tit.arg{\textbf{Semiparametric Inference for Fuzzy Regression-Discontinuity Designs}}

\begin{document}

\def\spacingset#1{\renewcommand{\baselinestretch}%
{#1}\small\normalsize} \spacingset{1}

\if0\blind
{
  \title{\bf Semiparametric Inference for Regression-Discontinuity Designs} 
  \author{Weiwei Jiang and Rong J.B. Zhu$^*$
  \hspace{.2cm}\\
  Institute of Science and Technology for Brain-inspired Intelligence\\
  Fudan University, Shangha China\\
  E-mail: rongzhu@fudan.edu.cn\\
  }
  \date{}
  \maketitle
}\fi

\if1\blind
{
  \bigskip
  \bigskip
  \bigskip
  \begin{center}
    {\LARGE\bf Semiparametric Inference for Regression-Discontinuity Designs}
  \end{center}
  \medskip
} \fi

\bigskip
\begin{abstract}
Treatment effects in regression discontinuity designs (RDDs) are often estimated using local regression methods. 
\cite{Hahn:01} demonstrated that the identification of the average treatment effect at the cutoff in RDDs relies on the unconfoundedness assumption and that, without this assumption, only the local average treatment effect at the cutoff can be identified. 
In this paper, we propose a semiparametric framework tailored for identifying the average treatment effect in RDDs, eliminating the need for the unconfoundedness assumption. 
Our approach globally conceptualizes the identification as a partially linear modeling problem, with the coefficient of a specified polynomial function of propensity score in the linear component capturing the average treatment effect.
This identification result underpins our semiparametric inference for RDDs, employing the $P$-spline method to approximate the nonparametric function and establishing a procedure for conducting inference within this framework.  
Through theoretical analysis, we demonstrate that our global approach achieves a faster convergence rate compared to the local method. 
Monte Carlo simulations further confirm that the proposed method consistently outperforms alternatives across various scenarios. 
Furthermore, applications to real-world datasets illustrate that our global approach can provide more reliable inference for practical problems.  
\end{abstract}

\vfill


\newpage
\spacingset{1.8} 

\pagenumbering{arabic}
\setcounter{page}{1}

\section{Introduction}
\label{sec:intro}

The regression discontinuity design (RDD) was initially introduced by \cite{TC:60} to examine the influence of merit awards on students' future academic outcomes. Since then, it has evolved into one of the most widely used strategies for estimating treatment effects across  disciplines such as economics, political science, biology, and medicine. 
In an RDD, units receive scores, and treatment allocation is determined by whether the score exceeds a predetermined cutoff value: units scoring above the cutoff are designated to the treatment condition, while those scoring below are assigned to the control condition. This treatment assignment rule creates a discontinuity in the probability of receiving treatment, allowing researchers to estimate the treatment effect by comparing units barely above and barely below the cutoff. 

The local linear or quadratic approximation \citep{Fan:92,FG:96} is a commonly employed approach in RDDs \citep{Hahn:01}. Specifically, researchers typically estimate the treatment effect locally around the cutoff, with the estimates depending on the selected bandwidth. 
Studies have investigated optimal bandwidth selection for local linear and quadratic regression specifications \citep{IK:12, CCT:14, CCT:15}. 
Another approach explored by researchers is the use of a global polynomial regression approach \citep{LL:10}. In practice, researchers often employ high-order polynomials (up to the fifth or sixth degree), with the polynomial degree selected using statistical information criteria or cross-validation. However, \cite{GI:19} argued against the use of high-order global polynomial approximations, instead advocating for inference based on local low-order polynomials, such as local linear or quadratic models. 
For comprehensive reviews of the RDD literature, refer to \cite{IL:08} and \cite{CIT:19}. 


\cite{Hahn:01} demonstrated that identifying the average treatment effect at the cutoff in RDDs relies on the unconfoundedness assumption. They also showed that, in the absence of this assumption, only the local average treatment effect at the cutoff can be identified. 
In this article, we delve into a semiparametric framework for RDDs, proposing a global approach that identifies the average treatment effect when the unconfoundedness assumption is violated. 
Our primary focus is on identifying the average treatment effect within a partially linear model framework. 
This entails capturing the treatment effect through the coefficient of a specified polynomial function of the assignment probability (or propensity score) in the linear component of the model.

This result, stemming from the unique identification structure,  underpins our semiparametric inference for RDDs, offering a more efficient estimation of the treatment effect at the cutoff compared to the earlier global polynomial regression approach \citep{LL:10}. 
By employing $P$-spline to approximate the nonparameteric component,  
we reformulate the semiparametric RDD framework into a regression based on a linear mixed-effects model.
We then derive the treatment effect estimate within this linear mixed-effects model framework. 
Theoretically, we establish the asymptotic normality of the estimate under regular conditions, highlighting the advantages of our approach, and propose an inference procedure. 
Finally, leveraging the asymptotic normality, we optimize the form of the polynomial function of the propensity score, which enhances the accuracy of our inference. 

Empirically, our experiments on simulated datasets support the assertion that our global approach provides a viable and superior  alternative to local estimation methods. 
Furthermore, we apply our global approach to evaluate the effect of antihypertensive treatment on reducing the risk of cardiovascular disease \citep{calonico2024regression}, and to assess the impact of incumbent advantage in U.S. House and Senate elections \citep{Lee:08,CFT:15}. Our analysis of these real-world datasets demonstrates that our global estimation approach yields more reliable inference. 

The rest of the article proceeds as follows. Section \ref{sec:setup} outlines the problem setting, and Section \ref{sec:identification} introduces a semiparametric framework for identifying the average treatment effect at the cutoff. Section \ref{sec:EI} presents estimation and inference procedures, demonstrating that both the bias and variance are approximately negligible compared to those of the local method. Section \ref{sec:simulation} shows the simulation results comparing our approach with the local method, while Section \ref{sec:empirical} evaluates its performance on real datasets.  
Concluding remarks are provided in Section \ref{sec:conclusion}. 
Proofs are gathered in the appendix. 

\section{Setup}
\label{sec:setup}

Throughout this paper, we consider the following problem setup. We have a random sample $(X_i, W_i, Y_i)$ of individuals that are independently distributed,  with the individuals in the sample indexed by $i=1,\cdots,n$. 
Using the potential outcome approach, let $(Y_i(0),Y_i(1))$ denote the pair of potential outcomes for individual $i$, 
and $W_i\in\{0,1\}$ represent the treatment received. 
Each individual $i$ is assigned only one treatment $W_i$, 
so we observe only the realized outcome $Y_i =Y_i(W_i)$.

Let $X_i$ be the running or forcing variable, and $c$ be the cutoff for $X_i$ at which the probability of treatment changes. 
Let $W_i^*$ be a dummy that indicates whether $X_i$ is above or below the threshold $c$, i.e., $W_i^*=\sbf{1}(X_i \geq c)$. 
Define the potential treatment status $W_i(w^*)$ as what an individual's treatment status would be if $W_i^*=w^*$.  
We have $W_i=W_i(1)W_i^*+W_i(0)(1-W_i^*)$. 
If $W_i^*=W_i$, we call it the \emph{sharp RDD}, where the probability of treatment assignment jumps from 0 to 1 when the forcing variable $X_i$ crosses the cutoff $x=c$. 
Otherwise, we call it the \emph{fuzzy RDD} which allows a smaller jump in the probability of treatment assignment at the cutoff and requires 
$$\lim\limits_{x\rightarrow c^-}p(x)\neq\lim\limits_{x\rightarrow c^+}p(x),$$
where $p(x)=\mathbb{P}(W_i=1 | X_i=x)$. Note that the propensity score $p(\cdot)$ is generally unknown and is estimated using the dataset. 

Our primary concern throughout this paper is the identification and inference under the fuzzy RDD. 
The sharp RDD, as a special case, can be applied with minor adjustments, which are detailed as needed. 

In fuzzy RDD settings, two causal quantities are of interest at the cutoff: average treatment effect (ATE) and local average treatment effect (LATE), defined as follows: 
\begin{align*}
\text{ATE} =& \mathbb{E}[Y_i(1)-Y_i(0) | X_i=c]\notag\\
\text{LATE} =& \mathbb{E}[Y_i(1)-Y_i(0) | X_i=c, W_i(0)<W_i(1)]. 
\end{align*}
Identifying the average treatment effect $\text{ATE}$ typically requires the unconfoundedness assumption. 
When this assumption does not hold, the LATE can be identified instead if the monotonicity assumption holds. 
However, the LATE only captures the treatment effect for compliers, defined as individuals having $W_i(0)<W_i(1)$. 
As a result, the LATE provides a limited perspective on the overall treatment effect, as it is restricted to the subset of compliers.

\section{Identification}
\label{sec:identification}

In this section, we propose a semiparametric framework that identifies the ATE without requiring the unconfoundedness assumption.

\subsection{Assumptions}
Define the conditional expectations of the potential outcomes $Y_i(1)$ and $Y_i(0)$, and the treatment assignment $W_i$, given $X_i$, as follows: 
\begin{align*}
\mu_1(x)=&\mathbb{E}[Y_i(1) \mid X_i=x], \notag\\
\mu_0(x)=&\mathbb{E}[Y_i(0) \mid X_i=x],  \notag \\
p(x) = & \mathbb{E}[W_i \mid X_i=x],
\end{align*} 
where $\mu_1(x)$ and $\mu_0(x)$ represent the expected outcomes under treatment and control, respectively, and $p(x)$ denotes the propensity score. 
Define the following deviations: 
\begin{align*}
    \epsilon_i(0) = & Y_i(0) - \mu_0(X_i), \\
    \epsilon_i(1) = & Y_i(1) - \mu_1(X_i), \\
    \epsilon_i = & W_i - p(X_i), 
\end{align*}
where $\epsilon_i(0)$ and $\epsilon_i(1)$ represent deviations of the potential outcomes from their conditional expectations, and $\epsilon_i$ denotes the deviation of the treatment assignment from the propensity score.
By definition, the expectations of $ \epsilon_i(0), \epsilon_i(1), \text{and} \ \epsilon_i $ are all zero and uncorrelated with $X_i$. 
It is worth noting that we do not assume $X_i$ to be independent of $(\epsilon_i(0), \epsilon_i(1), \epsilon_i )$.
The average treatment effect at the cutoff is defined as 
$$\tau_{c}=\mu_1(c)-\mu_0(c).$$

We impose the following assumptions: 
\begin{assumption} \label{assum-mu}
    $\mu_0(x)$ and $\mu_1(x)$ are both continuous at $x$. 
\end{assumption}

\begin{assumption} \label{assum-dis}
    $p(x)$ is continuous when $x \neq c$. At the cutoff, the right and left limits exist and are finite, with $\lim_{x \to c^+} p(x) =p(c) \neq \lim_{x \to c^-} p(x)$. 
\end{assumption}

\begin{assumption} \label{assum-ind}
    $\mathbb{E}[(\epsilon_i(1)-\epsilon_i(0)) \epsilon_i \mid X_i=x]$ is continuous when $x \neq c$. At the cutoff, the right and left limits exist and are finite, with $\lim_{x \to c^+} \mathbb{E}[(\epsilon_i(1)-\epsilon_i(0)) \epsilon_i \mid X_i=x] =\mathbb{E}[(\epsilon_i(1)-\epsilon_i(0)) \epsilon_i \mid X_i=c]$.  
\end{assumption}

Assumption \ref{assum-mu} imposes the standard requirement in RDDs that $\mu_1(x)$ and $\mu_0(x)$ are continuous over the entire domain of the running variable. 
Assumption \ref{assum-dis} reflects a fundamental condition in fuzzy RDDs, requiring the propensity score $p(x)$ to be discontinuous at the cutoff $x=c$ while remaining continuous elsewhere. Additionally, we assume that the propensity score is right-continuous at the cutoff. This is a minimal requirement, as in many applications, the propensity score between $c$ and $c+e$ is nearly identical when $e$ is very small. 
It is worth noting that Assumption \ref{assum-dis} holds for sharp RDDs by setting $p(x)=\sbf{1}(x\geq c)$.

Assumption \ref{assum-ind}, a novel and critical condition proposed in this paper, posits that the covariance between $\epsilon_i(1) - \epsilon_i(0)$ and $\epsilon_i$, conditional on $X_i=x$, exhibits continuity properties. 
Specifically, let $C(x)=\mathbb{E}[(\epsilon_i(1)-\epsilon_i(0)) \epsilon_i \mid X_i=x]$. Assumption \ref{assum-ind} requires that $C(x)$ is continuous when $x\neq c$ while allowing for a potential discontinuity at $x=c$ due to variations in confounding variables around the cutoff. 
This condition enhances the flexibility and applicability of our approach by accommodating settings where confounding factors differ near the cutoff.
Importantly, it does not require a discontinuity at the cutoff but permits its existence if warranted by the data. 

Let us take an example to illustrate this assumption. 
Consider a confounding variable $Z_i$ correlated with both $W_i$ and $(Y_i(0), Y_i(1))$. 
Suppose $\epsilon_i=b_wZ_i+\tilde{\epsilon}_i$ and $\epsilon_i(1)-\epsilon_i(0) =b_yZ_i+\breve{\epsilon}_i$,
where $b_w$ and $b_y$ are the coefficients, $\breve{\epsilon}_i \perp \!\!\! \perp \tilde{\epsilon}_i \mid X_i$, and $\mathbb{E}[Z_i \mid X_i]=0$. 
In this case, $C(x)=b_wb_y\text{Var}[Z_i \mid X_i=x]$. Assumption \ref{assum-ind} is satisfied when $\text{Var}[Z_i \mid X_i=x]$ is continuous when $x \neq c$ and right-continuous at $x=c$.  
This example illustrates that Assumption \ref{assum-ind} is quite mild, as it only requires the continuity of specific moments of the confounding variable conditional on $X_i=x$ for $x\neq c$, while permitting its discontinuity at the cutoff $x=c$.


\subsection{Identification}

We now present a semiparametric framework for identifying the average treatment effect at the cutoff in RDDs. 
\begin{theorem}\label{them-identification}
Suppose that Assumptions \ref{assum-mu}, \ref{assum-dis}, and \ref{assum-ind} hold. 
There exists a constant $\beta$ and a continuous function $f(x)$ such that 
\begin{align}\label{eqn:idt0}
Y_i= & \tau_{c}p(X_i) + \beta \sbf{1}(X_i \geq c) +f(X_i)+\varepsilon_i.
\end{align}
where $\varepsilon_i$ is uncorrelated with $X_i$ with expectation $0$ and variance $\sigma_i^2$.

In particular, the sharp RDD is a special case where the model in Equation \eqref{eqn:idt0} holds with $p(X_i) = \sbf{1}(X_i \geq c)$ and $\beta=0$.
\end{theorem}

Theorem \ref{them-identification} demonstrates that the average treatment effect at the cutoff $x=c$ can be expressed as the coefficient of the propensity score $p(X_i)$ in the linear component of the partial linear model \eqref{eqn:idt0}.

However, the identification of $\tau_c$ in the partially linear model \eqref{eqn:idt0} is infeasible without imposing any restriction on the nonparametric function $f(x)$ \citep{Robinson:88}. It is important to note that the identification condition in \cite{Robinson:88} does not hold in this model.   
To ensure the identification, we impose a structural restriction on $f(x)$. Specifically, we approximate $f(x)$ using a basis representation.  
This restriction facilitates the identification of $\tau_c$.

\begin{corollary} \label{them-identification-spline}
Assume that $f(x)$ lies within the span of the spline basis functions denoted by $\sbf{\phi}(x)$ and that that $\mathbb{E}[\sbf{\eta}_i\sbf{\eta}_i^{\top}]$ is non-singular, where $\sbf{\eta}_i=(p(X_i), \sbf{1}(X_i \geq c), \sbf{\phi}(X_i)^{\top})^{\top}$. 
Under these conditions, 
$\tau_c$ in the model \eqref{eqn:idt0} is identifiable. 
\end{corollary}

Corollary \ref{them-identification-spline} further demonstrates that the average treatment effect at the cutoff $x=c$ in RDDs is identified by the coefficient of the propensity score $p(X_i)$ in the linear component. 
The restriction corresponds to approximating $f(\cdot)$ with splines.
In the estimation step, we employ $P$-spline approach, which achieves a high level of accuracy, with the bias introduced by the restriction being negligible \citep{Ruppert:02,Ruppert:03,CKO:09}. We further discuss this in Section \ref{sec:second-stage}.

We further generalize in the following corollary that the coefficient of a specified polynomial function of the propensity score in the linear component captures the average treatment effect at the cutoff in fuzzy RDDs. 
\begin{corollary} \label{cor:identi}
    Consider a polynomial function $g(t) = a_1 t+a_2t^2+\cdots+a_mt^m$ for $t \in [0, 1]$, where $m \geq 1$. The coefficient vector of $g(t)$ is denoted as $\mbf{a} = (a_1, a_2, \cdots, a_m)^{\top}$, subject to $\mbf{a}^\top \mbf{G} \mbf{a} =1$ and $a_1 \geq 0$, where $\mbf{G}$ is a given positive semi-definite matrix with trace $\text{tr}(\mbf{G})=m$. 
    Suppose the conditions in Theorem \ref{them-identification} hold.  
    Then in fuzzy RDDs, there exists a constant $\beta$ and a continuous function $f(x)$ such that 
    \begin{align}\label{model-g}
        Y_i = \tau_c g(p(X_i)) + \beta \sbf{1}(X_i \geq c) + f(X_i) + \varepsilon_i. 
    \end{align}
    Furthermore, $\tau_c$ in \eqref{model-g} is identifiable under the conditions in Corollary \ref{them-identification-spline}. 
\end{corollary}

Corollary \ref{cor:identi} generalizes $p(X_i)$ to $g(p(X_i))$, where $g(\cdot)$ is a polynomial function. 
This generalization can enhance estimation by selecting an optimal $g(\cdot)$, 
as discussed in Section \ref{sec:choice-g}.  
The restriction $\mbf{a}^\top \mbf{G} \mbf{a} =1$ serves as a normalization to facilitate the determination of $\mbf{a}$. 
The restriction $a_1>0$ ensures consistency with the case that $g(\cdot)$ degenerates to the identical function when $m=1$.

\section{Estimation and Inference}
\label{sec:EI}
In this section, we first propose a two-stage estimator for the treatment effect at the cutoff within the semiparametric identification framework. We then establish its asymptotic normality for conducting inference and determine an optimal form of $g(\cdot)$. 

\subsection{First-stage estimation of propensity score}
In the first-stage estimation, we estimate $p(X_i) = \mathbb{E}[W_i \mid X_i]$, which is flexible and allows for a jump at the cutoff. 
To achieve this, we apply a nonparametric logistic regression \citep{HTF:11} with two segments, allowing for a jump at the cutoff.
Specifically, we approximate $p(X_i)$ by a linear combination of a set of basis functions within the following model:
\begin{align}\label{eqn:first-stage}
    \text{logit}\left(p(X_i)\right) = \sbf{\alpha}^{\top}_{0} \mbf{S}_{0}(X_i) + \sbf{\alpha}^{\top}_{1} \mbf{S}_{1}(X_i) \cdot \sbf{1}(X_i \geq c), 
\end{align}
where $\text{logit}(x)= \log\frac{x}{1-x}$. 
Here, for $j=0,1$, the vector of basis functions evaluated at $X_i$ is $$\mbf{S}_{j}(X_i) = \left(1, S_{j,1}(X_i), \cdots, S_{j,K_j}(X_i) \right)^\top, $$  
where $S_{j,k}(X_i)$ denotes the $k$-th basis function for $k=1, \cdots, K_j$, and $K_j$ represents the number of basis functions.
The vectors $\sbf{\alpha}_{0}$ and $\sbf{\alpha}_{1}$ are of dimensions $(K_0+1)$ and $(K_1+1)$, respectively.

We estimate $\sbf{\alpha}_{0}$ and $\sbf{\alpha}_{1}$ by fitting the logistic regression model, obtaining the propensity score estimate $\hat{p}(X_i)$ of $p(X_i)$ for $i=1,\cdots,n$. 
Once the propensity scores are estimated, the corresponding values of $g(\hat{p}(X_i))$ for $i=1,\cdots,n$ are subsequently computed. 

\subsection{Second-stage estimation of ATE} 
\label{sec:second-stage}
We employ the \emph{P}-spline approach to estimate $f(x)$. Specifically,
$f(x)$ is approximated using splines with a set of radial basis functions in the following form:
\begin{align} \label{eqn:f-splines}
    f(x) = \sum_{j=0}^{q} \beta_j x^j + \sum_{k=1}^K \gamma_k |x - \kappa_k|^{2q+1}, 
\end{align}
where $q$ denotes the polynomial and $\kappa_1 < \kappa_2 < \cdots < \kappa_K$ represent a set of knots along the running variable.
The radial basis functions are the default splines in the R package \emph{SemiPar} \citep{SemiPar:05}, as described in \cite{FKW:01} and \cite{Ruppert:03}. Specifically, the terms $1, x, \cdots, x^q$ form the polynomial, capturing global trends, while $|x - \kappa_1|^{2q+1}, \cdots, |x - \kappa_K|^{2q+1}$ are radial basis functions, representing local deviations at the specified knots $\kappa_1, \cdots, \kappa_K$. 
As the number of knots increases, the approximation is of a high level of accuracy. However, it may lead to severe overfitting. To mitigate this issue, we employ the $P$-spline approach to estimate $f(x)$ by imposing a penalty on the parameters $\gamma_k$, $k= 1, \cdots, K$, assuming that they are identically and independently distributed according to a normal distribution $\mathcal{N}(0, \sigma_\gamma^2)$. 

The choice of $K$ is a critical consideration in nonparametric estimation. \cite{Ruppert:02} provides an algorithm for selecting the optimal number of truncated polynomial basis functions by minimizing the generalized cross-validation. A simple alternative is to select $K$ as $\max \{n/4, 20\}$, ensuring that there are four or five points in each sub-interval consecutive knots \citep{Ruppert:02,Ruppert:03}. 

Denoting $\mbf{z} = ( |x- \kappa_1|^{2q+1}, \cdots, |x- \kappa_K|^{2q+1} )^\top$ and $\sbf{\gamma} = (\gamma_1, \cdots, \gamma_K)^\top$, Equation \eqref{eqn:f-splines} is reformulated as the following form:
\begin{align} \label{eqn:f-PL}
    f(x) = \beta_0 + \beta_1 x + \cdots + \beta_q x^q + \sbf{\gamma}^\top \mbf{z}, 
\end{align}
where $\sbf{\gamma} \sim \mathcal{N}(\sbf{0}_K, \sigma_\gamma^2 \mbf{I}_K)$. 
Let $\sbf{\beta}=( \beta, \beta_0, \cdots, \beta_q)^\top$, $\mbf{g} = (  g(\hat{p}(X_i)), \cdots, g(\hat{p}(X_n)))^\top$, $\mbf{x}=( \sbf{1}(x \geq c), 1, x, \cdots, x^q)^{\top}$, $\mbf{y} = (y_1, \cdots, y_n)^\top$, $\mbf{X} = (\mbf{x}_1, \cdots, \mbf{x}_n)^\top$, $\mbf{Z} = (\mbf{z}_1, \cdots, \mbf{z}_n)^\top$ and $\sbf{\varepsilon} = ({\varepsilon}_1, \cdots, {\varepsilon}_n)^\top$. 
From Equation \eqref{eqn:f-PL}, we rewrite the model \eqref{model-g} in Corollary \ref{cor:identi} as the following linear mixed-effects model:  
\begin{align} \label{model-pl}
\mbf{y} = \tau_c \mbf{g} + \sbf{\beta}^\top \mbf{X} + \sbf{\gamma}^\top \mbf{Z} + \sbf{\varepsilon}, 
\end{align}
where $(\sbf{\gamma}^\top, \sbf{\varepsilon}^\top  )^\top \sim \mathcal{N}(\sbf{0}_{K+n}, \text{diag}(\sigma_\gamma^2 \mbf{1}_K, \sigma_1^2, \cdots, \sigma_n^2 ))$. 
The model \eqref{model-pl} disregards the estimation error of $\hat{p}(X_i)$, which is negligible due to the consistency of nonparametric logistic regression \citep{GS:93}.

From the model, the covariance matrix of $\mbf{y}$ is given by 
$\mbf{V}_0 = \sigma_\gamma^2 \mbf{Z}\mbf{Z}^\top + \sigma^2 \text{diag}(\lambda_1^2, \cdots, \lambda_n^2)$, where $\lambda_i^2 = \sigma_i^2/\sigma^2$ for $i=1, \cdots, n$. 
However, to simplify the estimation step, we disregard the heteroscedasticity and instead use the simplified covariance term $\mbf{V} = \sigma_\gamma^2 \mbf{Z}\mbf{Z}^\top + \sigma^2 \mbf{I}_n$. 
Thus, we estimate $\sbf{\theta} = (\tau_c, \sbf{\beta}^\top)^\top$ by applying generalized least squares, yielding the following estimate: 
\begin{align}\label{theta-tilde}
    \tilde{\sbf{\theta}} = ( \mbf{U}^\top \mbf{V}^{-1} \mbf{U})^{-1}  \mbf{U}^\top \mbf{V}^{-1} \mbf{y}, 
\end{align}
where $\mbf{U} = (\mbf{g}, \mbf{X})$. 
It is clear that the estimate $\tilde{\sbf{\theta}}$ remains consistent and unbiased. 
We obtain the estimates of $\sigma_\gamma^2, \sigma^2$ by applying the maximum likelihood or restricted maximum likelihood method \citep{corbeil1976restricted}. 
Then, we calculate the estimated covariance matrix as $\hat{\mbf{V}} = \hat{\sigma}_\gamma^2 \mbf{Z}\mbf{Z}^\top + \hat{\sigma}^2 \mbf{I}_n $. 
Substituting these estimates to Equation \eqref{theta-tilde}, the estimate of $\sbf{\theta}$ is obtained as 
\begin{align*}
    \hat{\sbf{\theta}} = ( \mbf{U}^\top \hat{\mbf{V}}^{-1} \mbf{U})^{-1}  \mbf{U}^\top \hat{\mbf{V}}^{-1} \mbf{y}.
\end{align*}
Finally, the two-stage estimator of $\tau_c$ is given by 
\begin{align}\label{tau-hat}
    \hat{\tau}_c = \mbf{e}_1^\top ( \mbf{U}^\top \hat{\mbf{V}}^{-1} \mbf{U})^{-1}  \mbf{U}^\top \hat{\mbf{V}}^{-1} \mbf{y}, 
\end{align}
where $\mbf{e}_1$ is a $(q+3)$-dimensional vector whose first element is $1$ and others are $0$.

\subsection{Inference}
\label{sec:inference}

We present the asymptotic normality of the two-stage estimator $\hat{\tau}_c$ within the model \eqref{model-pl} for inference. 
Before introducing the property, we first impose the following assumptions. 
 \begin{assumption}\label{inference-condition}
 (a) $\text{rank}(U) = q+3$, where $\text{rank}(\cdot)$ denotes the matrix rank. \\
(b) $\lim_{n\rightarrow \infty} K = \infty$. \\
(c) $\lim_{n\rightarrow\infty} K/n=\rho \in [0, \infty)$. \\
(d) $\lim_{n\rightarrow\infty} [n-\text{rank}(\mbf{Z})]/n$ and $\lim_{n\rightarrow\infty} \text{rank}(\mbf{Z})/K$ exist and are positive.

 \end{assumption}

Part (a) in Assumption \ref{inference-condition} requires that $U$ has full column rank, which is a mild condition as $q$ is typically small, often taking values in $\{1,2,3\}$. Parts (b) and (c) impose constraints on the number of knots $K$, requiring that $K$ approaches infinity as $n\rightarrow \infty$, and that $K$ is of a smaller or the same order as $n$.
Part (d) allows the rank of the design matrix of $\mbf{Z}$ to be smaller but of the same order as $n$, and of the same order as $K$.

\begin{theorem}\label{them-inference}
Define 
$(q+3)\times (q+3)$ matrix $\mbf{J}$ given by 
$$\mbf{J}=\lim_{n\rightarrow\infty} n^{-1}(\mbf{U}^{\top}\mbf{V}^{-1}\mbf{U})(
\mbf{U}^{\top}\mbf{V}^{-1}\mbf{V}_0\mbf{V}^{-1}\mbf{U})^{-1}
(\mbf{U}^{\top}\mbf{V}^{-1}\mbf{U}).$$
Under the model \eqref{model-pl}, assuming that Assumption \ref{inference-condition} holds and that $\varepsilon_1,\cdots,\varepsilon_n$ are independent zero-mean Gaussian random variables with $\text{Var}(\varepsilon_i)=\sigma\lambda_i$, where $\lambda_i$ is known,   
we have that 
\begin{align*}
    \sqrt{n}(\hat{\sbf{\theta}}-\sbf{\theta}) \xrightarrow{d} \mathcal{N}(0,\mbf{J}^{-1}), \text{ as } n\rightarrow \infty.
\end{align*}
\end{theorem}

Theorem \ref{them-inference} follows the asymptotic normality of $\hat{\tau}_c$ with respect to $\tau_c$. 
\begin{corollary} \label{cor:inf-tau}
Let $\mbf{S} = \mbf{V}^{-1}(\mbf{I} - \mbf{H})$ and $\mbf{R}= \mbf{\Omega} - \mbf{\Omega} \mbf{H} - \mbf{H}^\top \mbf{\Omega}+\mbf{H}^\top \mbf{\Omega} \mbf{H}$, where $\mbf{\Omega} = \mbf{V}^{-1} \mbf{V}_0  \mbf{V}^{-1}$ and $\mbf{H} = \mbf{X}(\mbf{X}\mbf{V}^{-1}\mbf{X})^{-1}\mbf{X}^\top\mbf{V}^{-1}$.  
Under the conditions in Theorem \ref{them-inference}, we have that 
\begin{equation} \label{eqn:variance-tau}
    \frac{\hat{\tau}_c - \tau_c}{\sqrt{V_{\tau}}} \xrightarrow{d} \mathcal{N}(0, 1), \text{ as } n\rightarrow \infty, 
\end{equation}
where $$V_{\tau} = \mbf{g}^{\top}\mbf{R}\mbf{g}/ (\mbf{g}^{\top}\mbf{S}\mbf{g})^2$$.
\end{corollary}

Theorem \ref{them-inference} is proven in the case where the unknown variance components $\sigma_{\gamma}^2$ and $\sigma^2$ are estimated by maximum likelihood. If the restricted maximum likelihood method is employed instead, the above asymptotic normality still holds, as demonstrated by \cite{Jiang:98}. The condition that $\varepsilon_1,\cdots,\varepsilon_n$ are Gaussian is necessary to leverage the asymptotic results from \cite{Miller:77}. 
The condition that $\lambda_1,\cdots,\lambda_n$ are known simplifies the analysis and facilitates the application of the results in \cite{Miller:77}. 

We employ heteroskedasticity-consistent (HC) standard errors to account for heteroskedasticity, as suggested by \cite{white1980heteroskedasticity} and \cite{huang2022accounting}. 
Specifically, the variance $\mbf{V}_0$ is estimated as $\hat{\mbf{V}}_0 = \text{diag}( \hat{v}^2_1,\cdots, \hat{v}^2_n)$, where $\hat{v}_i = \frac{e_i}{1-h_i}$, with $e_i$ denoting the marginal residual and $h_i$ the leverage value for unit $i$.  

By substituting $\hat{\mbf{V}}$ and $\hat{\mbf{V}}_0$ into the definitions of $\mbf{R}$ and $\mbf{S}$ as stated in Corollary \ref{cor:inf-tau}, we obtain the estimates $\hat{\mbf{R}}$ and $\hat{\mbf{S}}$, respectively. 
Consequently, from Corollary \ref{cor:inf-tau}, an $1-\alpha$ confidence interval for $\tau_{c}$ is given by
\begin{align}\label{eqn:CI}
\left(\hat{\tau}_{c} - z_{\alpha/2}\hat{V}_{\tau}^{1/2}, \hat{\tau}_{c} + z_{\alpha/2}\hat{V}_{\tau}^{1/2}\right),
\end{align}
where $\hat{V}_{\tau} =  {\mbf{g}^{\top} \hat{\mbf{R}}\mbf{g} }/{ (\mbf{g}^{\top}\hat{\mbf{S}}\mbf{g})^2 }$ and $z_{\alpha/2}$ is the upper $\alpha/2$-quantile from a standard Gaussian distribution. 

\subsection{Comparison to the local approach}

Intuitively, our global approach utilizes all data points to estimate the treatment effect, whereas the local approach replies on data near the cutoff \citep{Hahn:01}. As a result, our global approach is expected to be more efficient than the local approach. 
We theoretically compare our global approach with the local approach and demonstrates that our approach exhibits lower bias and variance, achieving a faster convergence rate. 

Theorem \ref{them-inference} has demonstrated that 
the asymptotic variance of $\hat{\tau}_c$ is $O(n^{-1/2})$, while the asymptotic bias term is negligible, i.e., $o(n^{-1/2})$ as discussed in Section \ref{sec:second-stage}.
For local linear regression, \cite{FG:96} shows that the asymptotic bias is of order $h^2$, where $h$ is the bandwidth, and the asymptotic variance is of order $1/(nh)$. \cite{Hahn:01} discusses that these properties specifically apply to the local estimator in the context of RDDs. 
Typically, the optimal bandwidth is chosen as $h=O(n^{-1/5})$. 
Therefore, we have the following proposition. 
\begin{proposition}\label{prop:amse}
    Denote $B(\cdot)$, $V(\cdot)$ and $\text{MSE}(\cdot)$ as the asymptotic bias, asymptotic variance and asymptotic MSE, respectively. 
    Define $\hat{\tau}_{c}^{\text{local}}$ as the estimator corresponding to the local method, where the bandwidth is chosen as $h=O(n^{-1/5})$.
    Then we have 
    \begin{align*}
        \frac{B(\hat{\tau}_c)}{B(\hat{\tau}_{c}^{\text{local}})} = o(n^{-\frac{1}{10}}), \ \frac{V(\hat{\tau}_c)}{V(\hat{\tau}_{c}^{\text{local}})} = O(n^{-\frac{1}{5}}), \ \frac{\text{MSE}(\hat{\tau}_c)}{\text{MSE}(\hat{\tau}_{c}^{\text{local}})} = O(n^{-\frac{1}{5}}). 
    \end{align*}
\end{proposition}
Proposition \ref{prop:amse} demonstrates that our global approach is more effective than the local approach for RDDs. 

\subsection{Determining the form of \texorpdfstring{$g(\cdot)$}{g(.)}}
\label{sec:choice-g}

Until now, we have investigated the estimation and inference procedure of ${\tau}_c$ given a specified form of $g(\cdot)$. 
We now turn to the method for determining the optimal form of $g(\cdot)$ that minimizes the variance $V_{\tau}$. 

Recall that $g(t) = a_1t+\cdots+a_m t^m$. 
Let $\mbf{p}_k = (\hat{p}^k(X_1), \cdots, \hat{p}^k(X_n) )^\top$ for $k =1, \cdots, m$, and denote $\mbf{P} = (\mbf{p}_1, \cdots, \mbf{p}_m)$.   
We then express $\mbf{g}=(g(\hat{p}(X_1)),\cdots,g(\hat{p}(X_n)))^{\top}$ as $\mbf{g} = \mbf{P}\mbf{a}$. 
The form of $g(\cdot)$ is determined by finding a coefficient vector $\mbf{a}$ for the polynomial function $g(\cdot)$ that minimizes the variance $V_{\tau}$. 


Define $\mbf{Q}_{\text{R}} = m \mbf{P}^\top \mbf{R} \mbf{P} / \text{tr}(\mbf{P}^\top \mbf{R} \mbf{P})$ and $\mbf{Q}_{\text{S}} = m \mbf{P}^\top \mbf{S} \mbf{P} / \text{tr}(\mbf{P}^\top \mbf{S} \mbf{P})$. 
We specify $\mbf{G}$ in Corollary \ref{cor:identi} as $\mbf{Q}_{\text{S}}$, which implies that the optimal coefficient vector is determined under the conditions $\mbf{a}^\top \mbf{Q}_{\text{S}} \mbf{a} = 1$ and $a_1 \geq 0$. 
Therefore, based on Equation \eqref{eqn:variance-tau}, the variance minimization problem is equivalent to solve the following optimization problem:
\begin{align*}
    &\min_{\mbf{a}} \ \mbf{a}^\top \mbf{Q}_{\text{R}} \mbf{a}, 
    \text{ s.t.} \ \mbf{a}^\top \mbf{Q}_{\text{S}} \mbf{a} = 1 \ \text{and} \ a_1 > 0.   
\end{align*}

Note that $\mbf{Q}_{\text{S}}$ is positive semi-definite and might be singular. To address this issue, we first perform an orthogonal decomposition of it:
\begin{align*}
    \mbf{Q}_{\text{S}} = \begin{bmatrix} \mbf{S}_1 & \mbf{S}_2 \end{bmatrix} 
    \begin{bmatrix} \mbf{\Lambda} &  \\   &\mbf{0} \end{bmatrix} 
    \begin{bmatrix}
        \mbf{S}_1 \\ \mbf{S}_2 
    \end{bmatrix}, 
\end{align*}
where $\mbf{\Lambda}$ is a diagonal matrix containing all positive eigenvalues, and $\mbf{S}_1$ is the matrix consisting of the corresponding eigenvectors. 
In practice, eigenvalues below $10^{-5}$ are treated as zero. 
Let $\mbf{a} = \mbf{S}_1 \mbf{b}$ and then we just need to solve the following problem:
\begin{align} \label{prob:optb}
    &\min_{\mbf{b}} \ \mbf{b}^\top \mbf{S}^\top_1 \mbf{Q}_{\text{R}} \mbf{S}_1 \mbf{b}, 
    \text{ s.t.} \ \mbf{b}^\top \mbf{S}^\top_1 \mbf{Q}_{\text{S}} \mbf{S}_1 \mbf{b} = 1 \ \text{and} \ a_1 > 0.   
\end{align}

Since $\mbf{S}^\top_1 \mbf{Q}_{\text{S}} \mbf{S}_1$ is positive definite, 
Problem \eqref{prob:optb} is a classical optimization problem for minimizing a generalized Rayleigh quotient, with the optimal solution given by $\mbf{b}_{opt}$. 
The optimal coefficient vector $\mbf{a}_{opt} = \mbf{S}_1 \mbf{b}_{opt}$ is uniquely determined. 
Consequently, the optimal function $\mbf{g}$ is $\mbf{g}_{opt} = \mbf{P} \mbf{a}_{opt}$.


We conclude this section by discussing the choice of $m$. The form of $g(\cdot)$ that we use can be viewed as an approximation using a set of polynomial functions derived from $p(X_i)$. 
When $m$ is large, the search space for $g(\cdot)$ expands, which, however, also increases the variance. 
This is analogous to nonparametric fitting with a polynomial basis, where selecting a high-degree polynomial can lead to excessive variance. 
Thus, a practice approach is to select $m$ from a moderate range, such as $\{3,4,5,6,7\}$. 
In our calculations, we set $m=5$ and also evaluate performance across various $m$ values,
observing that the superiority of our approach is not sensitive to this choice. 
To optimally balance bias and variance, cross-validation can be employed.  

\section{Simulation experiments}
\label{sec:simulation}

In this section, we conduct Monte Carlo experiments to assess the performance of our method. 
We consider two scenarios: one where the unconfoundedness assumption holds and the other where it does not. 
Note that in this section we focus on the performance for fuzzy RDDs. However, we also include a simulation study for sharp RDDs and report the results in Table \ref{tab:sharp} of the Appendix \ref{app:results}, where our method demonstrates superior performance. 

We set the cutoff at 0 and generate the running variable $X_i$ from a uniform distribution, i.e. $X_i \sim \mathcal{U}(-1, 1)$. 
We consider three regression functions for the potential outcomes, labeled M1, M2, and M3, respectively:
\begin{align*} 
\text{M1:}\quad & \mathbb{E}[Y_i(0) \mid X_i=x] =3x^3,  
		\quad \mathbb{E}[Y_i(1) \mid X_i=x]= 4x^3; \\  
\text{M2:}\quad & \mathbb{E}[Y_i(0) \mid X_i=x] =0.42+0.84x+1.00x^2+e^{x/2} , \notag \\
		\quad & \mathbb{E}[Y_i(1) \mid X_i=x]=0.42+0.84x+1.00x^2+e^{x/2}+ x^2\sbf{1}(x \geq 0); \\  
\text{M3:}\quad & \mathbb{E}[Y_i(0) \mid X_i=x] = 0.48+1.27x+7.18x^2+20.21x^3+21.54x^4+7.33x^5, \notag \\
		\quad & \mathbb{E}[Y_i(1) \mid X_i=x]= 0.52+0.84x-3.00x^2+7.99x^3-9.01x^4+3.56x^5. 
\end{align*}
Each of the three models above represents a distinct pattern of change.
In M1, the average treatment effect increases monotonically throughout the running variable. In M2, the treatment effect remains 0 when the running variable is negative and increases for $x \geq 0$. M3 corresponds to the function in \cite{Lee:08}.  

Based on the setup outlined above, we now consider the two scenarios to generate the observed treatments $W_i$ and outcomes $Y_i$.  

\textbf{Scenario 1: UA holds.}
We first consider the scenario where the unconfoundedness assumption (UA) holds. 
The treatment variable is given as follows:
\begin{align} 
    \mathbb{P}(W_i =1 \mid X_i=x) = & \text{expit}\Big( 0.5x+0.2x^2 + 2\sbf{1}(x \geq 0) -1 \Big),\notag\\ 
    W_i \mid X_i= & x \sim \text{Bernoulli}\left(  \mathbb{P}(W_i =1 \mid X_i=x) \right), \label{eqn:sim-w} 
\end{align}
where $\text{expit}(x) = 1/(1+e^{-x})$.
The observed outcome for unit $i$ is generated as follows:
\begin{align} \label{eqn:sim-y-way1}
    Y_i = \mathbb{E}[Y_i(0) \mid X_i=x] + \mathbb{E}[Y_i(1) - Y_i(0) \mid X_i=x] W_i + \eta_i, 
\end{align}
where $\eta_i \sim \mathcal{N}(0, \sigma^2_{\eta})$. We set $\sigma^2_{\eta}$ to ensure that $R^2$ in Equation \eqref{eqn:sim-y-way1} is maintained as 0.75. 

\textbf{Scenario 2: UA is violated.}
We introduce an extra exogenous random error $\varepsilon_i$, which is defined as  $\varepsilon_i \sim \sbf{1}(X_i < 0)  \mathcal{N}(0, \sigma^2_\varepsilon) + \sbf{1}(X_i \geq 0)  \mathcal{N}(0, 2\sigma^2_\varepsilon)$, and add it to the propensity score model, as follows:
\begin{align*} 
     \mathbb{P}(W_i =1 \mid X_i=x) = \text{expit}\Big( 0.5x+0.2x^2 + 2\sbf{1}(x \geq 0) -1  + \varepsilon_i \Big). 
\end{align*}
We set the variance $\sigma^2_\varepsilon = \text{Var}\left(0.5X_i+0.2X_i^2 + 2\sbf{1}(X_i \geq 0) -1 \right)/3$.
We then generate $W_i$ according to Equation \eqref{eqn:sim-w} as well. 
$Y_i(0)$ and $Y_i(1)$ are generated as follows:
\begin{align}
    Y_i(0) & = \mathbb{E}[Y_i(0) \mid X_i] + \varepsilon_i(0), \label{eqn:sim-y0} \\ 
    Y_i(1) & = \mathbb{E}[Y_i(1) \mid X_i] + \varepsilon_i(1), \label{eqn:sim-y1}
\end{align}
where $\varepsilon_i(0) = c_0 \varepsilon_i$ and $\varepsilon_i(1) = c_1 \varepsilon_i$. Here, $c_0$ and $c_1$ are constants, determined by setting $R^2$ to 0.75 in Equations \eqref{eqn:sim-y0} and \ref{eqn:sim-y1}, respectively. The observed outcome is then generated by $Y_i = Y_i(0) + (Y_i(1) - Y_i(0))W_i$. 
Clearly, in this case, the unconfoundedness assumption is violated, while our Assumption 3 holds. 

In the first-stage regression, we use natural splines for $\mbf{S}_0(\cdot)$ and $\mbf{S}_1(\cdot)$ in Equation \eqref{eqn:first-stage}. Following the recommendation of \cite{harrell2001regression}, the number of knots is set to either 3 or 5, with the final selection based on the largest $R^2$ \citep{zhang2017coefficient}.  

In the second-stage regression, we use the R package \emph{SemiPar} \citep{SemiPar:05} to implement our global estimation. The degree $q$ is set to 1 in Equation \eqref{eqn:f-PL}, which results in basis functions corresponding to cubic thin plate splines.
The number and placement of knots are set by default in the \emph{SemiPar} package.
The parameter $m$ is chosen to be $5$, two degrees higher than the order of $f$ and natural splines, consistent with the configuration used in \cite{Lee:08}. 
We evaluate the performance of our method for different values of $m$ and report the results in Table \ref{tab:m} of the Appendix \ref{app:results}. The results demonstrate that our method maintains its superiority across varying values of $m$.
In the implementation of local estimators, the bandwidth is selected using two methods: those proposed by \cite{IK:12} and \cite{CCT:14}. A triangular kernel is employed in both cases.

Simulations are conducted for two sample sizes, 500 and 1000, with 10,000 Monte Carlo replications for each case. 
The results for the case where the unconfoundedness assumption holds are presented in Table \ref{tab:simulation-case1}, while the results for the case where the unconfoundedness assumption does not hold are shown in Table \ref{tab:simulation-case2}. 
We compare our method, denoted as ``PL" for the partially linear model, with two local estimators from \cite{IK:12} and \cite{CCT:14}, referred to as ``IK" and ``CCT", respectively.

From both tables, we make the following observations. 
First, our method consistently achieves significantly smaller RMSE compared with the local methods across all cases in both scenarios.
In terms of bias, our method shows much smaller values in M3 and comparable values in M1 and M2. 
Second, the empirical coverage of our method is close to nominal coverage, with generally smaller confidence interval lengths. 
Third, when comparing the two sample sizes, $n=500$ and $n=1000$, the advantage of our method over the local methods is more pronounced when the sample size is smaller. 
Finally, comparing the two scenarios, 
our method shows a larger advantage when the unconfoundedness assumption is violated than when it holds. 
In the scenario where the unconfoundedness assumption does not hold, both the bias and RMSE values of our method are consistently much smaller than those of the local methods across all cases.

\begin{table}[!ht] 
\caption{Performance of various methods when the unconfoundedness assumption holds}
\label{tab:simulation-case1}
\begin{threeparttable} 		
\begin{center}
\resizebox{0.9\columnwidth}{!}
{
\begin{tabular}{ c | c | c c c c  |  c c c c  }
\hline
\multirow{2}{*}{Model} & \multirow{2}{*}{Method}  &  \multicolumn{4}{c|}{$n=500$} & \multicolumn{4}{c}{$n=1000$}    \\ 	
\cline{3-10}
& & RMSE  & Bias & EC & ACL & RMSE  & Bias & EC & ACL \\ 			
\hline
\multirow{3}{*}{M1} & IK & 0.632 & -0.094 & 0.962 & 2.515 & 0.416 & 0.007 & 0.932 & 1.691 \\ 
& CCT & 0.794 & 0.075 & 0.983 & 3.750 & 0.483 & 0.052 & 0.966 & 1.928 \\ 
& PL & 0.056 & 0.028 & 0.962 & 0.235 & 0.074 & 0.013 & 0.943 & 0.180 \\    
\cline{1-10}
\multirow{3}{*}{M2} & IK & 0.391 & -0.135 & 0.938 & 1.281 & 0.287 & -0.101 & 0.907 & 0.911 \\ 
& CCT & 0.875 & -0.019 & 0.984 & 18.804 & 0.357 & -0.033 & 0.984 & 1.465 \\ 
& PL & 0.054 & -0.006 & 0.956 & 0.191 & 0.038 & -0.001 & 0.959 & 0.138 \\  
\cline{1-10}		
\multirow{3}{*}{M3} & IK  & 2.370 & 0.211 & 0.935 & 9.476 & 2.113 & -0.429 & 0.935 & 7.015 \\ 
& CCT & 4.006 & 0.884 & 0.961 & 15.512 & 2.625 & -0.455 & 0.974 & 9.368 \\ 
& PL & 0.715 & -0.110 & 0.955 & 3.856 & 0.671 & -0.068 & 0.967 & 2.066 \\  
\hline    					
\end{tabular}
}
\end{center}
\begin{tablenotes}
\item \small   ``RMSE": root mean square error, ``Bias": bias values, ``EC": empirical coverage rate and ``ACL": average $95\%$ confidence interval length. 
\end{tablenotes}
\end{threeparttable}
\end{table}

\begin{table}[!ht] 
\caption{Performance of various methods when the unconfoundedness assumption does not hold}
\label{tab:simulation-case2}
\begin{threeparttable} 		
\begin{center}
\resizebox{0.9\columnwidth}{!}
{
\begin{tabular}{ c | c | c c c c  |  c c c c  }
\hline
\multirow{2}{*}{Model} & \multirow{2}{*}{Method}  &  \multicolumn{4}{c|}{$n=500$} & \multicolumn{4}{c}{$n=1000$}    \\ 	
\cline{3-10}
& & RMSE  & Bias & EC & ACL & RMSE  & Bias & EC & ACL \\ 			
\hline
\multirow{3}{*}{M1} & IK & 0.800 & -0.016 & 0.947 & 3.426 & 0.657 & -0.126 & 0.930 & 2.544 \\ 
& CCT & 1.037 & -0.007 & 0.962 & 4.496 & 0.719 & -0.092 & 0.930 & 2.994 \\ 
& PL & 0.086 & -0.006 & 0.935 & 0.325 & 0.049 & 0.001 & 0.962 & 0.223 \\   
\cline{1-10}
\multirow{3}{*}{M2} & IK & 0.436 & 0.029 & 0.918 & 1.684 & 0.331 & -0.062 & 0.951 & 1.161 \\ 
& CCT & 0.964 & 0.293 & 0.984 & 4.368 & 0.582 & -0.048 & 0.973 & 2.004 \\ 
& PL & 0.058 & 0.004 & 0.934 & 0.254 & 0.045 & 0.002 & 0.967 & 0.162 \\
\cline{1-10}		
\multirow{3}{*}{M3} & IK & 7.309 & -3.860 & 0.910 & 21.338 & 4.675 & -2.927 & 0.926 & 14.520 \\ 
& CCT & 10.991 & -3.441 & 0.955 & 37.823 & 7.464 & -3.187 & 0.985 & 22.749 \\ 
& PL & 1.191 & 0.005 & 0.951 & 3.216 & 0.863 & 0.190 & 0.948 & 2.697 \\      
\hline
\end{tabular}
}
\end{center}
\begin{tablenotes}
\item \small   ``RMSE": root mean square error, ``Bias": bias values, ``EC": empirical coverage rate and ``ACL": average $95\%$ confidence interval length. 
\end{tablenotes}
\end{threeparttable}
\end{table}


\section{Empirical studies}
\label{sec:empirical}

In this section, we apply our global approach to study two problems: evaluating the effect of antihypertensive treatment on reducing the risk of cardiovascular disease \citep{calonico2024regression}, and assessing the effect of incumbent advantage in U.S. House and Senate elections \citep{Lee:08,CFT:15}. 

\subsection{The effect of antihypertensive treatment on reducing the risk of cardiovascular disease}
\cite{calonico2024regression} provided an insightful article in which they used the teaching version of the Framingham Heart Study dataset to explore whether antihypertensive medication could reduce the risk of cardiovascular disease, treating systolic blood pressure as the running variable. 
The Framingham Heart Study is a long-term prospective study on the etiology of cardiovascular disease among a population of free-living individuals in Framingham, Massachusetts \citep{tsao2015cohort}. The teaching version of the dataset is available upon request from the Biologic Specimen and Data Repository Information Coordinating Center. 

We follow the steps outlined in \cite{calonico2024regression} and continue to use the teaching version of the dataset to explore the treatment effect of antihypertensive medication on the occurrence of cardiovascular disease. However, we exclude all observations from examination cycles 1 and 2 to ensure the independence of the samples. 
Since the diagnostic criteria for hypertension are a systolic blood pressure greater than 140 mm Hg or a diastolic blood pressure greater than 90 mm Hg, we select these as the running variables, respectively.
The treatment variable indicates whether the antihypertensive medication was used during the exam. 

In our study, Cardiovascular Disease is defined as Angina Pectoris, Myocardial infarction (hospitalized and silent or unrecognized), coronary insufficiency (unstable Angina), or Fatal Coronary Heart Disease. 
Thus, the observed outcome is a binary variable which equals 1 if the condition is present during follow-up and 0 otherwise. 
The covariates include sex, age, the number of cigarettes smoked per day, BMI, heart rate (ventricular rate) in beats/min, and the presence of diabetes and prevalent coronary heart disease on exam. 
Diastolic blood pressure will serve as a covariate if systolic blood pressure is the running variable and vice versa.  
After excluding all missing values, the final study dataset includes 2,799 units. 
It is worth noting, however, that any findings should not be used to interpret actual results, as the data have been redacted for teaching purposes. 

Table \ref{tab:anti} presents the results, where ``Systolic Blood Pressure" indicates that systolic blood pressure is used as the running variable, and ``Diastolic Blood Pressure" conveys a similar meaning. 
Similar to the results from the local methods \citep{calonico2024regression}, the results from our approach are not statistically significant. 
However, compared with the local methods, our approach yields treatment effect estimates with much lower standard errors, suggesting that it may be more reliable.

 \begin{table}[!ht] 
    	\caption{Average treatment effect of antihypertensive medication on reducing the risk of Cardiovascular Disease}
        \label{tab:anti}
    	\begin{threeparttable} 		
    		\begin{center}
    			\resizebox{0.9\columnwidth}{!}
    			{
    				\begin{tabular}{ c | c c c c  |  c c c c  }
    					\hline
    				 \multirow{2}{*}{Method}  &  \multicolumn{4}{c|}{ Systolic Blood Pressure } & \multicolumn{4}{c}{Diastolic Blood Pressure}    \\ 	
    					\cline{2-9}
    					 & Estimate  & Std.error & $z$-value & $p$-value & Estimate  & Std.error & $z$-value & $p$-value \\ 
    					\hline	
                        IK & 1.136 & 3.044 & 0.373 & 0.709 & 1.285 & 7.253 & 0.177 & 0.859 \\ 
   CCT & 0.495 & 2.403 & 0.206 & 0.837 & 0.426 & 3.678 & 0.116 & 0.908 \\ 
   PL & -0.009 & 0.226 & -0.039 & 0.969 & -0.053 & 0.081 & -0.653 & 0.514 \\
   \hline
    				\end{tabular}
    			}
    		\end{center}
    	\end{threeparttable}
\end{table}

Moreover, we restrict the scope of Cardiovascular Disease to Hospitalized Myocardial Infarction, and re-evaluate the effect of antihypertensive treatment on reducing the risk of cardiovascular disease. The results are reported in Table \ref{tab:rob}. 
From the table, the results from our method are similar to those in Table \ref{tab:anti}, with no significant change and the same sign.
However, there is a substantial difference in magnitude and the direction reverses for the two local estimates.
Given the assumption that defining Cardiovascular Disease as Hospitalized Myocardial Infarction largely overlaps with its definition based on multiple measures, this study demonstrates that our global approach is more reliable and consistent.

\begin{table}[!ht] 
    	\caption{Results of treatment effect estimates when Cardiovascular Disease is defined as the single measure of  Hospitalized Myocardial Infarction}
        \label{tab:rob}
    	\begin{threeparttable} 		
    		\begin{center}
    			\resizebox{0.9\columnwidth}{!}
    			{
    				\begin{tabular}{ c | c c c c  |  c c c c  }
    					\hline
    				 \multirow{2}{*}{Method}  &  \multicolumn{4}{c|}{ Systolic Blood Pressure } & \multicolumn{4}{c}{Diastolic Blood Pressure}    \\ 	
    					\cline{2-9}
    					 & Estimate  & Std.error & $z$-value & $p$-value & Estimate  & Std.error & $z$-value & $p$-value \\ 
    					\hline	
                        IK & -0.100 & 0.738 & -0.135 & 0.893 & 0.986 & 1.289 & 0.765 & 0.444 \\ 
   CCT & -0.472 & 1.591 & -0.296 & 0.767 & -0.960 & 10.695 & -0.090 & 0.929 \\ 
   PL & -0.040 & 2.351 & -0.017 & 0.987 & -0.061 & 0.089 & -0.688 & 0.491 \\   
   \hline
    				\end{tabular}
    			}
    		\end{center}
    	\end{threeparttable}
\end{table}



\subsection{The effect of incumbent advantage}
We demonstrate the performance of our method in sharp RDDs by analyzing the effects of incumbent advantage in U.S. House and Senate elections by analyzing two datasets. 
The first dateset, consisting of 6558 samples, is sourced from U.S. House elections \citep{Lee:08}, 
while the other dataset, comprising 1390 samples, examines party-level advantage in U.S. Senate elections spanning from 1914 to 2010 \citep{CFT:15}.  
In these datasets, the forcing variable $x_i$ represents the margin of victory of the Democratic party in a given election, with the outcome variable being the Democratic vote share in the subsequent election. 

Our identification result can directly be applied sharp RDDs, and the estimation method is also applicable to sharp RDDs.
Unlike in fuzzy RDDs, there is no need to perform a first-stage regression or account for $g(\cdot)$. We utilize cubic thin plate splines for global estimation and compare our approach with the local estimators. 
The results are reported in Table \ref{tab:election}.
For both datasets, the PL method performs comparably to the local estimator, suggesting that our method could serve as a viable alternative option for sharp RDDs in practice.

    \begin{table}[!ht] 
    	\caption{Average treatment effect of incumbent advantage}
        \label{tab:election}
    	\begin{threeparttable} 		
    		\begin{center}
    			\resizebox{0.9\columnwidth}{!}
    			{
    				\begin{tabular}{ c | c c c c  |  c c c c  }
    					\hline
    				 \multirow{2}{*}{Method}  &  \multicolumn{4}{c|}{ U.S House elections } & \multicolumn{4}{c}{U.S. Senate elections}    \\ 	
    					\cline{2-9}
    					 & Estimate  & Std.error & $z$-value & $p$-value & Estimate  & Std.error & $z$-value & $p$-value \\ 
    					\hline	
                        IK & 0.080 & 0.008 & 9.571 & $<0.001$ & 0.066 & 0.010 & 6.456 & $<0.001$ \\ 
  CCT & 0.064 & 0.012 & 5.471 & $<0.001$ & 0.074 & 0.015 & 5.081 & $<0.001$ \\ 
  PL & 0.065 & 0.016 & 3.977 & $<0.001$ & 0.055 & 0.010 & 5.381  & $<0.001$\\  
   \hline
    				\end{tabular}
    			}
    		\end{center}
    	\end{threeparttable}
    \end{table}
We lack knowledge of the true treatment effects in the real datasets.
To further assess the performance of our method relative to local estimators, we conduct a simulated experiment by setting artificial cutoffs at $c=\pm 0.1$, retaining the original control and treated units, and generating corresponding virtual datasets. 
Using our global estimator and the local estimators on these virtual datasets, we infer the treatment effects at $c=\pm 0.1$ and summarize the results in Table \ref{tab:imaginary}. 

Under the conditions of the sharp RDD, the true estimands at the cutoffs at $c=\pm 0.1$ are $0$.
From Table \ref{tab:imaginary}, our method fails to reject the null hypothesis in all situations, while IK significantly rejects it when $c=0.1$ in the U.S. House election dataset. Additionally, CCT only fails to reject it when $c=-0.1$ in the U.S. Senate election dataset. 
Therefore, our method offers more reliable inferences in this context. 
This investigation underscores the potential advantage of our method in providing more reliable inferences compared to the local estimators for RDDs. 
  \begin{table}[!ht] 
     	\caption{Results from virtual datasets, where artificial cutoffs are set at $c=\pm 0.1$}
        \label{tab:imaginary}
    	\begin{threeparttable} 		
   		\begin{center}
 			\resizebox{0.9\columnwidth}{!}
  			{
 				\begin{tabular}{ c | c | c c c c  |  c c c c  }
 					\hline
				 \multirow{2}{*}{Cutoff} & \multirow{2}{*}{Method}  &  \multicolumn{4}{c|}{ U.S. House elections } & \multicolumn{4}{c}{U.S. Senate elections}    \\ 	
  					\cline{3-10}
  					& & Estimate  & Std.error & $z$-value & $p$-value & Estimate  & Std.error & $z$-value & $p$-value \\ 
					\hline	
                 \multirow{3}{*}{$c=0.1$} & IK & -0.024 & 0.010 & -2.466 & 0.014 & -0.018 & 0.012 & -1.475 & 0.140 \\ 
 & CCT & -0.025 & 0.012 & -2.107 & 0.035 & -0.032 & 0.016 & -1.971 & 0.049 \\ 
 & PL & -0.016 & 0.023 & -0.720 & 0.472 & -0.022 & 0.014 & -1.525 & 0.127 \\ 
  \hline

  \multirow{3}{*}{$c=-0.1$} & IK & -0.009 & 0.007 & -1.277 & 0.202 & -0.016 & 0.015 & -1.082 & 0.279 \\ 
 & CCT & -0.027 & 0.010 & -2.585 & 0.010 & -0.014 & 0.019 & -0.745 & 0.456 \\ 
  & PL & -0.027 & 0.020 & -1.342 & 0.180 & -0.010 & 0.016 & -0.630 & 0.528 \\  
  \hline
 				\end{tabular}
 			}
 		\end{center}
 	\end{threeparttable}
    \end{table}


\section{Conclusion}
\label{sec:conclusion}

The estimation of the average treatment effect at the cutoff is generally considered infeasible in RDDs without the unconfoundedness assumption; instead, only the local treatment effect is identified. 
In this paper, we have propose a semiparametric inference framework that identifies the average treatment without relying on the unconfoundedness assumption.

Moreover, this framework offers an efficient global estimator for the average treatment effect without replying on the unconfoundedness assumption. Compared to local estimators commonly used for RDDs, we have demonstrated the superiority of our global estimator through theoretical analysis and simulation experiments. 
We have also applied our proposed method to three datasets, yielding reliable inference.
We believe that our proposed global estimator is a viable alternative to local estimators for RDDs.

There are several potential extensions to consider. 
This paper focuses on the use of $P$-splines to represent the nonparametric part, but other nonparametric methods could also be efficient. Further investigation in this area is warranted.  
Another potential avenue is the regression kink design \citep{Card:15}, where \cite{CCT:14} and \cite{GI:19} recommend using local quadratic approach, building on the argument presented by \cite{Hahn:01}. 
Similar to our study on RDDs, our semiparametric inference framework could offer a robust alternative to local estimators in the regression kink design. 
Additionally, extending our framework to RDDs with covariates \citep{C:16} and RDDs with a discrete running variable \citep{KR:18} represents worthwhile directions for future research. 

\clearpage
\bibliographystyle{chicago}
\bibliography{AE}

\newpage
\appendix
\setcounter{figure}{0}
\setcounter{table}{0}
\renewcommand{\thelemma}{A.\arabic{lemma}}
\renewcommand{\theequation}{A.\arabic{equation}}
\renewcommand{\thefigure}{A.\arabic{figure}}
\renewcommand{\thetable}{A.\arabic{table}}
\section{Appendix}

\subsection{The proof of Theorem \ref{them-identification}}
Denoting $f_1(x)=\mu_0(x)(1-p(x))+(\mu_1(x)-\tau_{c})p(x)$, we have 
\begin{align}\label{eqn:plm2}
Y_i= & \tau_{c}p(X_i)+f_1(X_i)+\upsilon_i.
\end{align}
where $\upsilon_i=\epsilon_i(1)W_i+\epsilon_i(0)(1-W_i)+(W_i-p(X_i))(\mu_1(X_i)-\mu_0(X_i))$. 
We now show that $f_1(x)$ is continuous at $c$. 
From the definition of $f_1(x)$, 
\begin{align*}
f_1(c)= & p(c)(\mu_1(c)-\tau_{c})+(1-p(c))\mu_0(c)=\mu_0(c)\notag\\ 
\lim_{x\rightarrow c^{+}}f_1(x)= & \lim_{x\rightarrow c^{+}}p(x)(\mu_1(x)-\tau_{c})+(1-p(x))\mu_0(x)=\mu_0(c)\notag\\ 
\lim_{x\rightarrow c^{-}}f_1(x)= & \lim_{x\rightarrow c^{+}}p(x)(\mu_1(x)-\tau_{c})+(1-p(x))\mu_0(x)=\mu_0(c),
\end{align*}
where the last steps in both equations above holds on the basis of Assumption \ref{assum-dis} and the definition of $\tau_c$. 
Thus, $f_1(x)$ is continuous at $c$.

Next, we investigate the term $\upsilon_i$.
Denote $C(x)=\mathbb{E}[(\epsilon_i(1)-\epsilon_i(0))\epsilon_i \mid X_i=x]$.
Noting $\mathbb{E}[\epsilon_i(0) \mid X_i] = \mathbb{E}[\epsilon_i(1) \mid X_i] = 0$, we have 
\begin{align*}
\mathbb{E}[\upsilon_i \mid X_i] = & \mathbb{E}[\epsilon_i(1)W_i+\epsilon_i(0)(1-W_i) \mid X_i]\notag\\
 = & \mathbb{E}[\epsilon_i(0) \mid X_i]+ \mathbb{E}[(\epsilon_i(1)-\epsilon_i(0)) p(X_i) \mid X_i] + \mathbb{E}[(\epsilon_i(1)-\epsilon_i(0))\epsilon_i \mid X_i] \notag\\
 = & C(X_i) 
\end{align*}
Denote $C(c^-) = \lim_{x \to c^-} C(x)$. 
From Assumption \ref{assum-ind}, $C(x)$ can be written as 
\begin{align}\label{eqn:h}
    C(x)=f_2(x)+\beta \sbf{1}(x \geq c), 
\end{align}
where $\beta = C(c) - C(c^-)$ and $f_2(x) = C(x) \sbf{1}(x < c) + [C(x) -\beta] \sbf{1}(x \geq c)$ is a continuous function. 

Denote $\varepsilon_i = \upsilon_i - C(X_i)$ and $f(x) = f_1(x) + f_2(x)$. 
From inserting $\varepsilon_i$ and $f(x)$ into Equation \eqref{eqn:plm2}, we have  
\begin{equation*}
    Y_i=  \tau_{c}p(X_i)+ \beta \sbf{1}(X_i \geq c) + f(X_i) + \varepsilon_i, 
\end{equation*}
where $f(x)$ is a continuous function. 
It is easy to check that $\varepsilon_i$ is not correlated with $X_i$ and $\mathbb{E}[\varepsilon_i]=0$ by definition. 

In the sharp RDD, we have $p(X_i) = \sbf{1}(X_i \geq c)$ and $\epsilon_i=0$, implying that $\mathbb{E}[(\epsilon_i(1) - \epsilon_i(0))\epsilon_i \mid X_i=x] =0$. 

\subsection{The proof of Corollary \ref{them-identification-spline}}
The assumption that $f(x)$ is in the span of the spline basis functions denoted by $\sbf{\phi}(X)$ implies that $f(x)$ is approximated using splines represented by a set of 
radial basis functions, as shown in the following form:
\begin{align} \label{eqn:f-splines-1}
    f(x) = \sbf{\beta}_b^\top\sbf{\phi}(x), 
\end{align}
where $\sbf{\beta}_b$ is a coefficient vector of $\sbf{\phi}(x)$.
It follows that 
\begin{equation*}
    Y_i=  \tau_{c}p(X_i)+ \beta \sbf{1}(X_i \geq c) + \sbf{\beta}_b^\top\sbf{\phi}(X_i) + \varepsilon_i,
\end{equation*}
which forms a linear regression model. 
Therefore, $\tau_c$ is identifiable if the regressors in the regression model are linearly independent, which is guaranteed by the assumption $\mathbb{E}[\sbf{\eta}_i\sbf{\eta}_i^{\top}]$ is non-singular.

\subsection{The proof of Corollary \ref{cor:identi}}
We first introduce the following lemma. 
\begin{lemma} \label{lem:collinearity}
    Suppose that $g(t) = a_1t + \cdots+a_m t^m$ with $m \geq 1$ and $t \in [0, 1]$. Denote $\mbf{a} = (a_1, \cdots, a_m)^\top$. If $ \mbf{a}^\top \mbf{G} \mbf{a} =1$ and $a_1 > 0$, where $ \mbf{G}  $ is a positive semi-definite matrix with trace $\text{tr}(\mbf{G}) = m$, then there does not exist constants $\lambda_1 \neq 0$ and $\lambda_2$, such that $t-g(t) = \lambda_1 g(t) + \lambda_2$ for any $t$. 
\end{lemma}
\begin{proof}
    Assume that there exists a constant $\lambda_1 \neq 0$ and $\lambda_2$, such that $t-g(t) = \lambda_1 g(t) + \lambda_2$ for any t. 
    Note that $\lambda_1 \neq -1$ . We have 
    \begin{align*}
        g(t) = \frac{1}{1+\lambda_1} t - \frac{\lambda_2}{1+\lambda_1}. 
    \end{align*}
Based on the form of $g(\cdot)$, we have $\lambda_2 = 0$ and $\mbf{G} = 1$, implying $| a_1 | =1$.  
    Note that $|\frac{1}{1+\lambda_1}|$ equals to $1$ only when $\lambda_1=-2$. 
    However, when $\lambda_1 =-2$, $a_1 = \frac{1}{1+\lambda_1} \leq 0$. 
    Contradiction!
    
\end{proof}

From Theorem \ref{them-identification}, we write the Equation \eqref{eqn:idt0} as follows:
    \begin{align*}
        Y_i & = \tau_c p(X_i) + \beta \sbf{1}(X_i \geq c) + f(X_i) + \varepsilon_i \\
        & = \tau_c g(p(X_i)) + \tau_c \Big( p(X_i) - g(p(X_i)) \Big) + \beta \sbf{1}(X_i \geq c) +  f(X_i)   + \varepsilon_i.
    \end{align*}
Similar to Equation \eqref{eqn:h}, there exists a constant $\alpha$ and a continuous function $f_3(x)$ such that 
$$ p(x)- g(p(x)) = f_3(x) + \alpha \sbf{1}(X_i \geq c). $$
Thus, we have 
\begin{align*}
    Y_i & = \tau_c g(p(X_i)) +  (\beta + \tau_c \alpha) \sbf{1}(X_i \geq c) +  (f(X_i)+f_3(X_i))   + \varepsilon_i. 
\end{align*}  
Note that in the above model, $p(X_i) - g(p(X_i))$ is not collinear with $g(p(X_i))$, as shown in Lemma \ref{lem:collinearity}. 
Therefore, $\tau_c$ is still identifiable under the conditions in Corollary \ref{them-identification-spline}.

\subsection{The proof of Theorem \ref{them-inference}}
Let $\nu_1=n-K_r$, $\nu_2=K_r$, where $K_r=\text{rank}(\mbf{Z})$. 
We define the matrices 
$\mbf{G}_1 = \sigma^2 \text{diag}(\lambda^2_1, \cdots, \lambda^2_n)$ and $\mbf{G}_2=\sigma_{\gamma}^2\mbf{Z}\mbf{Z}^{\top}$,  
and introduce the $2\times 2$ matrix $\mbf{C}$ with its $(i,j)$-element defined as follows: 
$$[\mbf{C}]_{ij}=\lim_{n\rightarrow\infty}\frac{\text{tr}[\mbf{V}^{-1}\mbf{G}_i\mbf{V}^{-1}\mbf{G}_j]}{\sqrt{\nu_i\nu_j}}.$$


\cite{Miller:77} established the asymptotic properties of maximum likelihood estimates $\hat{\sbf{\theta}}$, in a general mixed-effects model under the following conditions: \\
A1: $K\rightarrow\infty$ as $n\rightarrow \infty$,\\
A2: $\lim_{n\rightarrow\infty}K/n=\rho \in [0, \infty]$,\\
A3: $\lim_{n\rightarrow\infty} (n-K_r)/n$ and $\lim_{n\rightarrow\infty} K_r/K$ exist and are positive,\\
A4: $\mbf{J}$ is positive definite,\\
A5: $\mbf{C}$ is positive definite. \\
Theorem 3.1 of \cite{Miller:77} demonstrates that under conditions A1-A5, as $n\rightarrow\infty$, 
  $$\sqrt{n}(\hat{\sbf{\theta}}-\sbf{\theta}) \xrightarrow{d} \mathcal{N}(0, \mbf{J}^{-1} ).$$
In the general model setting in \cite{Miller:77}, it is assumed that the model errors are homogeneous. Although our model includes heterogeneous errors, due to the assumption that the matrix $\text{diag}(\lambda^2_1, \cdots, \lambda^2_n)$ is known, only the parameter $\sigma^2$ in $\mbf{G}_1$ is estimated by using the MLE method. Consequently, Theorem 3.1 of \cite{Miller:77} applies straightforwardly to our model setting. 

Theorem \ref{them-inference} is an application of Theorem 3.1 of \cite{Miller:77}. 
It is sufficient to verify that Conditions A1, 
A2, and A3 as stated in Theorem \ref{them-inference} are satisfied. 
It remains to demonstrate the fulfillment of Conditions A4 and A5.

First, we analyze the $(q+3)\times (q+3)$ matrix 
$$\mbf{J}=\lim_{n\rightarrow\infty} (\mbf{U}^{\top}\mbf{V}^{-1}\mbf{U})(
\mbf{U}^{\top}\mbf{V}^{-1}\mbf{V}_0\mbf{V}^{-1}\mbf{U})^{-1}
(\mbf{U}^{\top}\mbf{V}^{-1}\mbf{U})/n.$$
Note that $\mbf{V}$ and $\mbf{V}_0$ are positive definite. 
Because $\mbf{U}$ is a column full rank matrix based on Assumption \ref{inference-condition} (a), 
 $\mbf{U}^{\top}\mbf{V}^{-1}\mbf{U}$ and $\mbf{U}^{\top}\mbf{V}^{-1}\mbf{V}_0\mbf{V}^{-1}\mbf{U}$ are positive definite. 
Since for any $\mbf{t} \neq \mbf{0}$, 
\begin{align*}
&\mbf{t}^\top \mbf{J} \mbf{t} = \mbf{t}^\top \mathbb{E}\Big[(\mbf{U}^{\top}\mbf{V}^{-1}\mbf{U})(
\mbf{U}^{\top}\mbf{V}^{-1}\mbf{V}_0\mbf{V}^{-1}\mbf{U})^{-1}
(\mbf{U}^{\top}\mbf{V}^{-1}\mbf{U}) \Big] \mbf{t} \\
=& \mathbb{E} \Big[ \Big( (
\mbf{U}^{\top}\mbf{V}^{-1}\mbf{V}_0\mbf{V}^{-1}\mbf{U})^{-\frac{1}{2}}
(\mbf{U}^{\top}\mbf{V}^{-1}\mbf{U}) \mbf{t} \Big)^\top \Big( (
\mbf{U}^{\top}\mbf{V}^{-1}\mbf{V}_0\mbf{V}^{-1}\mbf{U})^{-\frac{1}{2}}
(\mbf{U}^{\top}\mbf{V}^{-1}\mbf{U}) \mbf{t}  \Big) \Big] > 0.  
\end{align*} 
Thus, $\mbf{J}$ is positive definite and Condition A4 holds.

Next, Let us check Condition A5.  
By definition, we just need to prove, for any $\mbf{t} = (t_1, t_2)^\top \neq \mbf{0}$, 
\begin{equation} \label{eqn:condi-A5}
    [\mbf{C}]_{11}t_1^2 - 2 [\mbf{C}]_{12}t_1t_2 + [\mbf{C}]_{22} t_2^2 >0. 
\end{equation}
Since $\mbf{G}_1$ and $\mbf{G}_2$ are both invertible and symmetric, we have $\mbf{V}^{-1} = \mbf{G}^{-1}_1 (\mbf{G}^{-1}_1 + \mbf{G}^{-1}_2)^{-1} \mbf{G}^{-1}_2 = \mbf{G}^{-1}_2 (\mbf{G}^{-1}_1 + \mbf{G}^{-1}_2)^{-1} \mbf{G}^{-1}_1$. 
Hence, the elements of $\mbf{C}$ are represented by 
\begin{align}
    [\mbf{C}]_{11} = &\lim_{n\rightarrow\infty} \frac{ \text{tr}[ \mbf{G}^{-1}_2 (\mbf{G}^{-1}_1 + \mbf{G}^{-1}_2)^{-1} \mbf{G}^{-1}_2 (\mbf{G}^{-1}_1 + \mbf{G}^{-1}_2)^{-1}  ]  }{\nu_1 } \notag \\ 
    = &\lim_{n\rightarrow\infty} \frac{ \text{tr}[ \mbf{G}^{-1}_2 (\mbf{G}^{-1}_1 + \mbf{G}^{-1}_2)^{-2} \mbf{G}^{-1}_2  ]  }{\nu_1 }, \label{eqn:condi-C11} \\
    [\mbf{C}]_{22} = &\lim_{n\rightarrow\infty} \frac{ \text{tr}[ \mbf{G}^{-1}_1 (\mbf{G}^{-1}_1 + \mbf{G}^{-1}_2)^{-2} \mbf{G}^{-1}_1  ]  }{\nu_2 }, \label{eqn:condi-C22} \\
    [\mbf{C}]_{12} 
    = &\lim_{n\rightarrow\infty} \frac{ \text{tr}[ \mbf{G}^{-1}_2 (\mbf{G}^{-1}_1 + \mbf{G}^{-1}_2)^{-2} \mbf{G}^{-1}_1  ]  }{ \sqrt{\nu_1 \nu_2} } \label{eqn:condi-C12}, 
\end{align}
where the last equality in Equation \eqref{eqn:condi-C11} holds since $\mbf{G}^{-1}_1$, $(\mbf{G}^{-1}_1 + \mbf{G}^{-1}_2)^{-1}$ and $\mbf{G}^{-1}_2$ are all symmetric. 

We plug Equations \eqref{eqn:condi-C11}, \eqref{eqn:condi-C22} and \eqref{eqn:condi-C12} into Equation \eqref{eqn:condi-A5}, and then 
\begin{align*}
    &[\mbf{C}]_{11}t_1^2 - 2 [\mbf{C}]_{12}t_1t_2 + [\mbf{C}]_{22} t_2^2 \notag \\ 
    = &\lim_{n\rightarrow\infty} \text{tr} \Big[   ( \frac{t_1}{ \sqrt{\nu_1}} \mbf{G}^{-1}_1 - \frac{t_2}{ \sqrt{\nu_2}} \mbf{G}^{-1}_2)^\top (\mbf{G}^{-1}_1 + \mbf{G}^{-1}_2)^{-2} ( \frac{t_1}{ \sqrt{\nu_1}} \mbf{G}^{-1}_1 - \frac{t_2}{ \sqrt{\nu_2}} \mbf{G}^{-1}_2)       \Big] >0, 
\end{align*}
where the last inequality holds given that the trace of a positive definite matrix is positive.

\subsection{The proof of Corollary \ref{cor:inf-tau}}

Based on the above proof, the asymptotic variance of $\hat{\tau}_c$ equals to $\mbf{e}^\top_1 \mbf{J}^{-1} \mbf{e}_1$. 
Denote $\mbf{H} = \mbf{X}(\mbf{X}\mbf{V}^{-1}\mbf{X})^{-1}\mbf{X}^\top\mbf{V}^{-1}$, $\mbf{S} = \mbf{V}^{-1}(\mbf{I} - \mbf{H})$ and $s = (\mbf{g}^\top \mbf{S} \mbf{g})^{-1}$. 
Then we have
\begin{align} \label{eqn:block-matrix-J}
(\mbf{U}^\top \mbf{V}^{-1} \mbf{U})^{-1} & = 
    \begin{bmatrix} \mbf{g}^\top \mbf{V}^{-1}\mbf{g} & \mbf{g}^\top \mbf{V}^{-1}\mbf{X} \\ \mbf{X}^\top \mbf{V}^{-1}\mbf{g} & \mbf{X}^\top \mbf{V}^{-1}\mbf{X} \end{bmatrix}^{-1} = \begin{bmatrix}
        s & \mbf{M}_{12} \\
        \mbf{M}_{12}^{\top} & \mbf{M}_{22}
    \end{bmatrix}, 
\end{align}
where $\mbf{M}_{12} = -s \mbf{g}^\top \mbf{V}^{-1} \mbf{X} (\mbf{X}^\top \mbf{V}^{-1}\mbf{X})^{-1}$
and $\mbf{M}_{22} = (\mbf{X}^\top \mbf{V}^{-1}\mbf{X})^{-1} + s^{-1} \mbf{M}_{12}^{\top}\mbf{M}_{12}$. 
Note that we have 
\begin{align}\label{eqn:Vmid}
    \mbf{U}^\top \mbf{V}^{-1} \mbf{V}_0 \mbf{V}^{-1} \mbf{U} = 
    \begin{bmatrix}
       \mbf{g}^\top \mbf{\Omega} \mbf{g} & \mbf{g}^\top \mbf{\Omega}\mbf{X} \\ \mbf{X}^\top \mbf{\Omega}\mbf{g} & \mbf{X}^\top \mbf{\Omega}\mbf{X} 
    \end{bmatrix}, 
\end{align}
where $\mbf{\Omega} = \mbf{V}^{-1} \mbf{V}_0  \mbf{V}^{-1}$. 

Denote $\mbf{R}= \mbf{\Omega} - \mbf{\Omega} \mbf{H} - \mbf{H}^\top \mbf{\Omega}+\mbf{H}^\top \mbf{\Omega} \mbf{H}$. 
From \eqref{eqn:block-matrix-J} \& \eqref{eqn:Vmid}, we have 
\begin{align*}
    V_{\tau} & = \mbf{e}^\top_1 (\mbf{U}^{\top}\mbf{V}^{-1}\mbf{U})^{-1}(
\mbf{U}^{\top}\mbf{V}^{-1}\mbf{V}_0\mbf{V}^{-1}\mbf{U})
(\mbf{U}^{\top}\mbf{V}^{-1}\mbf{U})^{-1} \mbf{e}_1 =  \frac{\mbf{g}^\top \mbf{R} \mbf{g}}{(\mbf{g}^\top \mbf{S} \mbf{g})^2}. 
\end{align*}

\subsection{Additional mmpirical studies}
\label{app:results}
\begin{table}[ht]
\caption{The impact of \texorpdfstring{$m$}{m} on our global estimator, with results based on the M1 model and a sample size of $n = 500$.}
\label{tab:m}
\centering
\begin{tabular}{c|cccc}
  \hline
 $m$ & RMSE & Bias & CR & ACL \\ 
  \hline
2 & 0.186 & -0.020 & 0.981 & 0.486 \\
3 & 0.297 & -0.038 & 0.971 & 1.368 \\
4 & 0.060 & 0.003 & 0.976 & 0.249 \\
5 & 0.056 & 0.028 & 0.962 & 0.235 \\
6 & 0.267 & -0.015 & 0.943 & 0.870 \\
7 & 0.301 & -0.082 & 0.962 & 1.234 \\
   \hline
 IK & 0.632 & -0.094 & 0.962 & 2.515 \\ 
 CCT & 0.794 & 0.075 & 0.983 & 3.750 \\ 
   \hline
\end{tabular}
\end{table}

    \begin{table}[!ht] 
		\caption{Performance of three methods in the sharp RDD, 
        where the data generating process follows Scenario 2 of the fuzzy RDD in Section \ref{sec:simulation}, except that the assignment probabilities are defined as $P(X_i)=1$ for $X_i\geq c$ and $P(X_i)=0$ for $X_i< c$. 
        }
		\label{tab:sharp}
		\begin{threeparttable} 		
			\begin{center}
				\resizebox{0.9\columnwidth}{!}
				{
					\begin{tabular}{ c | c | c c c c  |  c c c c  }
						\hline
							\multirow{2}{*}{Model} & \multirow{2}{*}{Method}  &  \multicolumn{4}{c|}{500} & \multicolumn{4}{c}{1000}    \\ 	
						\cline{3-10}
						 & & RMSE  & Bias & EC & ACL & RMSE  & Bias & EC & ACL \\ 			
						\hline
						
						 \multirow{3}{*}{M1} & IK & 0.267 & -0.043 & 0.916 & 0.953 & 0.195 & -0.032 & 0.925 & 0.696 \\ 
   & CCT & 0.331 & -0.042 & 0.901 & 1.111 & 0.222 & -0.020 & 0.940 & 0.772 \\ 
   & PL & 0.235 & -0.034 & 0.972 & 1.171 & 0.173 & -0.020 & 0.974 & 0.890 \\  
						\hline
						
						 \multirow{3}{*}{M2} & IK & 0.159 & -0.059 & 0.889 & 0.500 & 0.117 & -0.054 & 0.882 & 0.379 \\ 
   & CCT & 0.198 & -0.003 & 0.935 & 0.775 & 0.155 & -0.032 & 0.922 & 0.560 \\ 
   & PL & 0.152 & -0.045 & 0.969 & 0.685 & 0.114 & -0.041 & 0.965 & 0.541 \\
						\hline		
						
						 \multirow{3}{*}{M3} & IK & 1.233 & 0.007 & 0.941 & 4.591 & 0.975 & 0.073 & 0.915 & 3.320 \\ 
   & CCT & 2.197 & 0.058 & 0.948 & 8.103 & 1.645 & 0.242 & 0.932 & 5.734 \\ 
   & PL & 1.060 & 0.074 & 0.948 & 4.123 & 0.791 & 0.089 & 0.935 & 2.955 \\
						\hline

					\end{tabular}
				}
			\end{center}
			\begin{tablenotes}
				\item \small   ``RMSE": root mean square error, ``Bias": bias values, ``EC": empirical coverage rate and ``ACL": average $95\%$ confidence interval length. 
			\end{tablenotes}
		\end{threeparttable}
	\end{table}

\end{document}